\documentclass[10pt,letterpaper,conference]{IEEEtran}
\usepackage{geometry}                % See geometry.pdf to learn the layout options. There are lots.
\usepackage{fullpage}
\usepackage{graphicx}
\usepackage{amssymb}
\usepackage{amsmath}%
\usepackage{amsfonts}%
\usepackage{amssymb}%
\usepackage{amsthm}
\usepackage{graphicx}
\usepackage{colonequals}
\usepackage{psfrag}
\usepackage{subfigure}
\usepackage{mathrsfs}
\usepackage{cite}
\usepackage{bbm}
\newtheorem{thm}{Theorem}

\newtheorem{cor}{Corollary}

\theoremstyle{definition}

\newtheorem{defn}{Definition}
\newtheorem{example}{Example}
\newtheorem{remark}{Remark}
\newcommand{\PR}{\mbox{Pr}}

\newcommand{\bS}{\mathbf{S}}
\newcommand{\bs}{\mathbf{s}}

\newcommand{\calU}{\mathcal{U}}

\newcommand{\calY}{\mathcal{Y}}
\newcommand{\calS}{\mathcal{S}}

\title{Privacy Against Statistical Inference}
\author{\IEEEauthorblockN{Fl\'avio du Pin Calmon}
\IEEEauthorblockA{
Massachusetts Institute of Technology\\
Cambridge, MA 02139\\
Email: flavio@mit.edu}
\and
\IEEEauthorblockN{Nadia Fawaz}
\IEEEauthorblockA{Technicolor\\
Palo Alto, CA 94301\\
Email: nadia.fawaz@technicolor.com}}
\begin{document}

\maketitle

\begin{abstract}
We propose a general statistical inference framework to capture the privacy
threat incurred by a user that releases data to a passive but curious adversary,
given utility constraints. We show that applying this general framework to the
setting where the adversary uses the self-information cost function naturally
leads to a non-asymptotic information-theoretic approach for characterizing the
best achievable privacy subject to utility constraints. Based on these results
we introduce two privacy metrics, namely average information leakage and maximum
information leakage. We prove that under both metrics the resulting design
problem of finding the optimal mapping from the user's data to a
privacy-preserving output can be cast as a modified rate-distortion problem
which, in turn, can be formulated as a convex program. Finally, we compare our
framework with differential privacy.
\end{abstract}

%\subsection{}
%\vspace{-1in}

\section{Introduction}

\subsection{Motivation}
Increasing volumes of user data are being collected over wired and wireless
networks, by a large number of companies who mine this data to provide
personalized services or targeted advertising to users. As a consequence,
privacy is gaining ground as a major topic in the social, legal, and business
realms. This trend has spurred recent research in the area of theoretical models
for privacy, and their application to the design of privacy-preserving services.
Most privacy-preserving techniques, such as anonymization, k-anonymity
\cite{Sweeney-2002} and differential privacy \cite{Dwork-McSherry-2006}, are
based on some form of perturbation of the data, either before or after the data
is used in some computation. These perturbation techniques provide privacy
guarantees at the expense of a loss of accuracy in the computation result, which
leads to a privacy-accuracy trade-off.

In this paper, we consider the general setting where a user wishes to release a
set of measurements to an analyst who provides a  service (e.g. a
recommendation system), while keeping data that are correlated with
these measurements private.  On one hand, the analyst is a legitimate receiver for these
measurements, from which he expects to derive some utility. On the other hand,
the correlation of these measurements with the user's private data gives the
analyst the ability to illegitimately infer private information. The tension
between the privacy requirements of the user and the utility expectations of the
analyst gives rise to the problems of privacy-utility trade-off modeling, and
the design of release schemes minimizing the privacy risks incurred by the user,
while satisfying the utility constraints of the analyst.

\subsection{Contributions}

%Our contributions are three-fold.
%
%
%First, we propose a general statistical inference framework to model the privacy threat incurred by the user: we model the privacy risk for the user as a cost gain for the analyst upon observation of the released output.
%
%We cast the design of privacy-preserving released output achieving an optimal privacy-accuracy trade-off as an optimization problem:  find the released output distribution which
%minimize the amount of information an adversary can infer about the private data by observing the released output, subject to a utility constraint on the released output.
%
%Show that for log-cost,
%
%- info theory shows up naturally. We can develop a non-asymptotic information-theoretic framework to characterize the information leaked (on average or at most) subject to a utility constraint
%
%- we can cast the problem into a convex optimization problem, which can be solved efficiently. Yields an optimal privacy scheme, minimizing the average information leakage or the maximum information leakage under a distortion constraint.
%
%We draw a comparison with other privacy notions.
%We show in particular that differential privacy does not prevent arbitrary information leakage.

Our contributions are three-fold.  First, we propose a general statistical
inference framework to capture the privacy threat incurred by a user who
releases information given certain utility constraints. The privacy risk  is
modeled as an inference cost gain by a passive but curious adversary upon
observing the information released by the user. In broad terms, this cost gain
represents the ``amount of knowledge'' learned by an adversary about the private
data after observing the user's output. The design problem of finding the
optimal mapping from the user's information to a privacy-preserving output is
formulated as an optimization problem where the cost gain of the adversary is
minimized for a given set of utility constraints. This formulation is general
and given in terms of minimizing both the average and the maximum cost gain,
being applicable to different cost functions.

Second, we apply this general framework to the case when the adversary uses the
self-information cost function. We show how this naturally leads to a
non-asymptotic  information-theoretic framework  to characterize the information
leakage subject to utility constraints. Based on these results we introduce two
privacy metrics, namely \textit{average information leakage} and \textit{maximum
information leakage}. We also demonstrate that the problem of designing a
privacy preserving mechanism that achieves the optimal privacy-accuracy tradeoff
both for the average and maximum information leakage can be cast as modified
rate-distortion problems. We then prove that these problems, in turn, can be
expressed as  convex programs. As a consequence, the privacy preserving mapping
that achieves the optimal privacy-utility tradeoff can be efficiently found
using convex minimization algorithms or widely available convex solvers.

Finally, we compare the average information leakage and maximum information
leakage metrics with differential privacy. We show that differential privacy
does not provide in general \textit{any} privacy guarantees in terms of average
or maximum information leakage. Furthermore, we introduce the definition of
\textit{information privacy}, and prove that information privacy implies both
differential privacy and privacy in terms of (average or maximum) information
leakage.

\subsection{Related Work}

In the privacy research community, a prevalent and strong notion of privacy is
that of differential privacy \cite{Dwork-McSherry-2006,dwork_differential_2006}.
Differential privacy bounds the variation of the distribution of the released
output given the input database, when the input database varies slightly, e.g.
by a single entry. Intuitively, released outputs satisfying differential privacy
render the distinction between "neighboring" databases difficult. % it hard to
distinguish between.  However, differential privacy neither provides guarantees,
nor an intuition, on the amount of information leaked when a differentially
private release occurs. Moreover, user data usually presents correlations.
Differential privacy does not factor in correlations in  user data, as the
distribution of user data is not taken into account in this model.  A natural
question is how the notion of privacy proposed in this paper compares to that of
differential privacy. We cover this question in more details in
Section~\ref{sec:compare}.

Several approaches rely on information-theoretic tools to model privacy-accuracy
trade-offs, such as
\cite{Reed-1973,Yamamoto-ITtrans1983,evfimievski_limiting_2003,Sankar-Poor-arxiv2010}.
Indeed, information theory, and more specifically rate-distortion theory, appear
as natural frameworks to analyze the privacy-accuracy trade-off resulting from
the distortion of correlated data. Although the approach we introduce in this
paper involves information theoretic metrics, it is fundamentally different from
previous information theoretic privacy models. Indeed, traditional information
theoretic privacy models, such as
\cite{Yamamoto-ITtrans1983,Sankar-Poor-arxiv2010}, focus on collective privacy
for all or subsets of the entries of a database, and provide asymptotic
guarantees on the average remaining uncertainty per database entry -- or
equivocation per input variable -- after the output release. More precisely, the
average equivocation per entry is modeled as the conditional entropy of the
input variables given the released output, normalized by the number of input
variables. In contrast, the general framework introduced in this paper provides
privacy guarantees in terms of bounds on the inference cost gain that an
adversary achieves by observing the released output. The use of a
self-information cost yields a non-asymptotic information theoretic framework
modeling the privacy risk in terms of information leakage. This framework, in
turn, can be used to design practical privacy preserving mappings.  Finally, we
would like to point out that the formulation in \cite{Reed-1973}, differs from
previously mentioned information theoretic models, and addresses a particular
case of the general framework introduced in this paper.
%rate-distortion-equivocation region,

The paper is organized as follows. We describe the set-up and the threat model
in Section~\ref{sec:genSet}, and  formulate the privacy-accuracy trade-off in
Section~\ref{sec:Challenge}. Our main results and their proofs are presented in
Section~\ref{sec:Results}. Finally, in Section~\ref{sec:compare} we draw a
comparison between the privacy notion proposed in this paper, and other existing
privacy models, leading to the concluding remarks in
Section~\ref{sec:Conclusion}.

\section{General Setup and threat model}
\label{sec:genSet}

In this section we outline the general  setup  considered in this paper and  the corresponding   threat model.

\subsection{General setup}

 We assume that there are two parties that communicate over a noiseless channel,
 namely Alice and Bob. Alice has access to a set of measurement points,
 represented by the variable $Y\in \mathcal{Y}$, that she wishes to transmit to
 Bob. At the same time, Alice requires that  a set of variables $S\in \calS$
 should remain private, where $S$ is jointly distributed with $Y$ according to
 the distribution $(Y,S)\sim p_{Y,S}(y,s)$, $(y,s)\in \mathcal{Y}\times
 \mathcal{S}$. Depending on the considered setting, the variable $S$ can be
 either directly accessible to Alice or inferred from  $Y$. If no privacy
 mechanism was in place, Alice would simply transmit $Y$ to Bob.

Bob has a utility requirement for the information sent by Alice. Furthermore,
Bob is honest but curious, and will try to learn $S$ from Alice's transmission.
Alice's goal is to find and transmit a distorted version of $Y$, denoted by
$U\in \mathcal{U}$, such that $U$ satisfies a target utility constraint  for
Bob, but ``protects'' (in a sense made more precise later) the private variable
$S$. We assume that Bob is passive but computationally unbounded, and will try
to infer $S$ based on $U$.

We consider, without loss of generality, that $S\rightarrow Y \rightarrow U$.
Note that this model can capture the case where $S$ is directly accessible by
Alice by appropriately adjusting the alphabet $\calY$. For example, this can be
done by representing $S\rightarrow Y$ as an injective mapping or allowing $\calS
\subset \calY$. In other words, even though the privacy mechanism is designed as
a mapping from $\calY$ to $\calU$, it is not limited to an output perturbation,
and it encompasses input perturbation settings.

\begin{defn} A privacy preserving mapping is a probabilistic mapping
  $g:\mathcal{Y}\rightarrow \mathcal{U}$ characterized by a transition
  probability $p_{U|Y}(u|y),~y\in \mathcal{Y},~u\in\mathcal{U}$.  \end{defn}
Since the framework developed here results in formulations that are similar to
the ones found in rate-distortion theory, we will use the term distortion to
indicate a measure of utility. Furthermore, we will use the terms utility and
accuracy interchangeably throughout the paper.  \begin{defn} Let $d:
  \mathcal{Y}\times \mathcal{U}\rightarrow \mathbb{R}^+$ be a given distortion
  metric. We say that a privacy preserving mapping has distortion $\Delta$ if
  $\mathbb{E}_{Y,U}[d(Y,U)]\leq \Delta$.  \end{defn}

We make the following assumptions: 
\begin{enumerate} 
    \item Alice and Bob know the prior distribution of $p_{Y,S}(\cdot)$. This
      represents the side information that an adversary has.  
  
    \item Bob has complete knowledge of the privacy preserving mapping, i.e.,
      $g$ and $p_{U|Y}(\cdot)$ are known.       
\end{enumerate}

    Note that this represents the
      \textit{worst-case} statistical side information that an adversary can
      have about the input.  
  \subsection{Threat model} We assume that Bob selects a revised distribution
  $q\in \mathcal{P}_{S}$, where $\mathcal{P}_{S}$ is the set of all probability
  distributions over $\calS$, in order to minimize an expected cost $C(S,q)$. In
  other words, the adversary chooses  $q$ as the solution of the minimization
  \begin{equation} c_0^*=\min_{q\in \mathcal{P}_{S}} \mathbb{E}_{S}[C(S,q)]
  \end{equation} prior to  observing $U$, and \begin{equation} c_u^*=\min_{q\in
    \mathcal{P}_{S}} \mathbb{E}_{S|U}[C(S,q)|U=u] \end{equation} after observing
    the  output $U$. Note that this restriction on Bob models a very broad class
    of adversaries that perform statistical inference, capturing how an
    adversary acts in order to infer a revised belief distribution over the
    private variables $S$ when observing $U$. After choosing this distribution,
    the adversary can perform an estimate of the input distribution (e.g. using
    a MAP estimator). However, the quality of the inference is inherently tied
    to the revised distribution $q$.

The average cost gain by an adversary after observing the output is
\begin{equation} \Delta C= c_0^*-\mathbb{E}_U[c_u^*].  \end{equation} The
  maximum cost gain by an adversary is measured in terms of the most informative
  output (i.e. the output that give the largest gain in cost), given by
  \begin{equation} \Delta C^*= c_0^*-\min_{u\in \mathcal{U}} c_u^*.
    \label{eq:maxcostgen} \end{equation}

In the next section we present a formulation for the privacy-accuracy tradeoff
based on this general setting.

\section{A general formulation for the privacy-accuracy tradeoff}\label{sec:Challenge}

\subsection{ The privacy-accuracy tradeoff as an optimization problem}

Our goal is to design privacy preserving mappings that minimize $\Delta C$ or
$\Delta C^*$ for a given distortion level $\Delta$, characterizing the
fundamental privacy-utility tradeoff. More precisely, our focus is to solve
optimization problems  over $p_{U|Y}\in\mathcal{P}_{{U|Y}}$ of the form
\begin{align}
  &\min~\Delta C \mbox{ or }\Delta C^* \label{eq:minDeltaC} \\
  &\mbox{s. t.~ } \mathbb{E}_{Y,U}[d(Y,U)]\leq \Delta~,
\end{align}
 where $\mathcal{P}_{{U|Y}}$ is the set of all conditional  probability
 distributions of $U$ given $Y$.

\begin{remark} In the remainder of the paper we consider only one distortion
  constraint. However, it is straightforward to generalize the formulation and
  the subsequent optimization problems to  multiple distinct distortion
  constraints $ \mathbb{E}_{Y,U}[d_1(Y,U)]\leq
  \Delta_1,\dots,\mathbb{E}_{Y,U}[d_n(Y,U)]\leq \Delta_n$. This can be done by
  simply adding an additional linear constraint to the convex program.
\end{remark}

\subsection{Application examples}

We illustrate next how the proposed model can be cast in terms of privacy
preserving queries and hiding features within data sets.

\subsubsection{Privacy-preserving queries to a database}

The framework described above can be applied to database privacy problems, such
as those considered in differential privacy. In this case we denote the private
variable as a vector $\bS=S_1,\dots, S_n$, where $S_j \in \mathcal{S}$, $1\leq j
\leq n$ and $S_1,\dots,S_n$ are discrete entries of a database that represent,
for example, the  entries of $n$ users. A (not necessarily
deterministic) function $f:\mathcal{S}^n \rightarrow \mathcal{Y}$ is calculated
over the database with output $Y$ such that $Y=f(S_1,\dots,S_n)$. The goal of
the privacy preserving mapping is to present a query output $U$ such that the
individual entries $S_1,\dots,S_n$ are ``hidden'', i.e. the estimation cost gain
of an adversary is minimized according to the previous discussion, while still
preserving the utility of the query in terms of the target distortion
constraint. We illustrate this case with the counting query, which will be a
recurring example throughout the rest of this paper.

\begin{example}[Counting query]  Let $S_1,\dots,S_n$ be entries in a database,
  and define: \label{examp:query1}
\begin{equation}
\label{eq:count_query}
Y = f(S_1,\dots,S_n)=\sum_{i=1}^{n} \mathbbm{1}_A (S_i),
\end{equation}
where
\begin{equation*}
\mathbbm{1}_A(x) = \left\{
     \begin{array}{ll}
       1& \mbox{if~} x~\mbox{has property $A$,}\\
       0 & \mbox{otherwise.}
     \end{array}
   \right.
\end{equation*}
In this case there are two possible approaches: (i) output perturbation, where
$Y$ is distorted directly to produce $U$, and (ii) input perturbation, where
each individual entry $S_i$ is distorted directly, resulting in a new query
output $U$.
\end{example}

\subsubsection{Hiding dataset features}  Another important particularization  of
the proposed framework is   the obfuscation of  a set of features $S$ by
distorting the entries of a data set $Y$. In this case $|\calS|\ll |\calY|$,
and $S$ represents a set of features that might be inferred from the data $Y$,
such as age group or salary. The distortion can be defined according to the the
utility of a given statistical learning algorithm (e.g. a recommendation system)
used by Bob.

\section{Privacy-accuracy tradeoff results}\label{sec:Results}

The formulation introduced in the previous section is general and can be applied
to different cost functions. In this section we particularize the formulation to
the case where the adversary uses the self-information cost function, as
discussed below.

\subsection{The self-information cost function}
The \textit{self information} (or
\textit{log-loss}) cost function is given by
\begin{equation}
C(S,q)=-\log q(S).
\end{equation}
There are several motivations for using such a cost function. For an overview of
the central role of the self-information cost function  in prediction, we refer
the reader to \cite{merhav_universal_1998}.  %Perhaps the most important
%motivation is that the square root of the log-loss cost gain upper bounds the gain when general
%loss functions are used (under reasonable continuity restrictions).
% In addition,
Briefly, the self-information cost function is the only local, proper and smooth cost
function for an alphabet of size at least three. Furthermore, since the minimum
self-information loss probability assignments are essentially ML estimates, this
cost function is consistent with a ``rational'' adversary. In addition, the average cost-gain when using the self-information cost can be related to the cost gain when using any other bounded cost function  \cite{merhav_universal_1998}. Finally, as we will
see below, this minimization implies a ``closeness'' constraint between the
prior and a posteriori probability distributions in terms of KL-divergence. In Section \ref{sec:compare} we compare the resulting privacy measure with that of differential privacy and \textit{information-privacy}.

In the next sections we show how the cost minimization problems in
\eqref{eq:minDeltaC} used with the self-information cost function can be cast as
convex programs and, therefore, can be efficiently solved using interior point
methods or widely available convex solvers.

\subsection{Average information leakage}
It is straightforward to show that for the log-loss function $c_0^*=H(S)$ and, consequently,
$c^*_u=H(S|U=u)$, and, therefore
\begin{align}
  \Delta C=I(S;U)&=\mathbb{E}_U[D(p_{S|U}||p_{S})],
 % &=H(S)-\mathbb{E}_U[H(S|U=u)],
\end{align}
where $D(\cdot||\cdot)$ is the KL-divergence. The minimization
\eqref{eq:minDeltaC} can the be rewritten according to the following definition.

\begin{defn}
 The \textit{average information leakage} of a set of features $S$ given a
 privacy preserving output $U$ is given by $I(S;U)$. A privacy-preserving mapping
 $p_{U|Y}(\cdot)$ is said to provide the \textit{minimum average information
 leakage} for a distortion constraint $\Delta$  if it is the solution of the
 minimization
 \begin{align}
  \min_{p_{U|Y}}~&I(S;U) \label{eq:minAvgI}\\
  \mbox{s.t.~ }& \mathbb{E}_{Y,U}[d(Y,U)]\leq \Delta~.
  \end{align}
\end{defn}
Observe that finding the mapping $p_{U|Y}(u|y)$ that provides the minimum information leakage  is a  modified rate-distortion problem. Alternatively, we can
rewrite this optimization as
\begin{align}
  \min_{p_{U|Y}}~&\mathbb{E}_U[D(p_{S|U}||p_{S})]  \label{eq:minD} \\
  \mbox{s.t.~ } &\mathbb{E}_{Y,U}[d(Y,U)]\leq \Delta~.
\end{align}

The minimization \eqref{eq:minD} has an  interesting and intuitive
interpretation. If we consider KL-divergence as a metric for the distance
between two distributions,  \eqref{eq:minD} states that the revised distribution
after observing $U$ should be as close as possible to the a priori distribution in terms of KL-divergence.

The following theorem shows how the the optimization in the previous definition can be expressed as a convex optimization problem. We note that this optimization is solved in terms of the unknowns $p_{U|Y}(\cdot|\cdot)$ and $p_{U|S}(\cdot|\cdot)$, which are coupled together through a linear equality constraint.

\begin{thm}
\label{prop:avgInf}
Given $p_{S,Y}(\cdot,\cdot)$, a distortion function $d(\cdot,\cdot)$ and a distortion constraint $\Delta$, the  mapping  $p_{U|Y}(\cdot|\cdot)$ that minimizes the average information leakage  can be found by solving the following convex optimization (assuming the usual simplex constraints on the probability distributions):
\begin{align} 
\min_{p_{U|Y},p_{U|S}}& \sum_{u\in \calU}\sum_{ s \in \calS} p_{U|S}(u|s)p_S(s) \log \frac{ p_{U|S}(u|s)}{p_U(u)} \label{eq:obj_avg}\\
\mbox{\normalfont s.t.~ }& \sum_{u\in\calU} \sum_{ y \in \calY} p_{U|Y}(u|y)p_Y(y)d(u,y)\leq \Delta,\\
& \sum_{y \in \calY} p_{Y|S}(y|s)p_{U|Y}(u|y) = p_{U|S}{(u|s)} ~\forall u,s, \label{eq:constr_avg}\\
&  \sum_{s \in \calS} p_{U|S}(u|s)p_{S}(s) = p_{U}{(u)}~\forall u.
\end{align}
\end{thm}

\begin{proof}
Clearly the previous optimization is the same as \eqref{eq:minAvgI}. To prove the convexity of the objective function, note that
$h(x,a) = ax \log  x$
is convex for a fixed $a\geq 0$ and $x\geq 0$, and, therefore, the perspective of $g_1(x,z,a) = ax\log(x/z)$ is also convex  in $x$ and $z$ for $z>0,a\geq 0$ \cite{boyd_convex_2004}. Since the objective function \eqref{eq:obj_avg} can be written as
\begin{equation*}
 \sum_{u\in \calU}\sum_{ s \in \calS} g(p_{U|S}(u|s),p_U(u),p_S(s) ),
\end{equation*}
it follows the optimization is convex. In addition, since $p(u)\rightarrow 0 \Leftrightarrow p(u|s)\rightarrow 0 ~\forall u$,  the minimization is well defined over the probability simplex.
\end{proof}

\begin{remark}
Note that the previous optimization can also be solved using a dual minimization procedure analogous to the Arimoto-Blahut algorithm \cite{cover_elements_2006} by starting at a fixed marginal probability $p_U(u)$, solving a convex minimization at each step (with an added linear constraint compared to the original algorithm) and updating the marginal distribution. However, the above formulation allows the use of efficient algorithms for solving convex problems, such as interior-point methods. In fact, the previous minimization can be simplified to formulate the traditional rate-distortion problem as a single convex program, not requiring the use of the Arimoto-Blahut algorithm.
\end{remark}

\begin{remark}
The formulation in Theorem \ref{prop:avgInf} can be easily extended to the case when $U$ is determined directly from $S$, i.e. when Alice has access to $S$ and the privacy preserving mapping is given by $p_{U|S}(\cdot|\cdot)$ directly. For this, constraint \eqref{eq:constr_avg} should be substituted by
\begin{equation}
\label{eq:new_constr1}
 \sum_{y \in \calY} p_{Y|S}(y|s)p_{U|Y,S}(u|y,s) = p_{U|S}{(u|s)} ~\forall u,s,
\end{equation}
and the following linear constraint added
\begin{equation}
\label{eq:new_constr2}
 \sum_{s \in \calS} p_{S|Y}(s|y)p_{U|Y,S}(u|y,s) = p_{U|Y}{(u|y)} ~\forall u,y,
\end{equation}
with the minimization being performed over the variables $p_{U|Y,S}(u|y,s) ,p_{U|Y}(u|y)$ and  $p_{U|S}(u|s)$, with the usual simplex constraints on the probabilities.

\end{remark}

We now particularize the previous result for the case where $Y$ is a deterministic function of $S$.

%\section{Deterministic functions and additive noise}
\begin{cor}
\label{avg:determ}
If $Y$ is a deterministic function of $S$ and $S \rightarrow Y \rightarrow U$ then the minimization in \eqref{eq:minAvgI} can be simplified to a rate-distortion problem:
\begin{align}
  &\min_{p_{U|Y}}~I(Y;U)\\
  &\mbox{\normalfont s. t.~ } \mathbb{E}_{Y,U}[d(Y,U)]\leq D~.
\end{align}
Furthermore, by restricting $U=Y+Z$ and $d(Y,U)=d(Y-U)$, the optimization reduces to
\begin{align}
  &\max_{p_Z}~H(Z)\\
  &\mbox{s. t.~ } \mathbb{E}_{Z}[d(Z)]\leq \Delta~.
\end{align}
\end{cor}
\begin{proof}
Since $Y$ s a deterministic function of $S$ and  $S\rightarrow Y \rightarrow U$, then
\begin{align}
I(S;U)&=I(S,Y;U)-I(Y;U|S)\\
          &=I(Y;U)+I(S;U|Y)-I(Y;U|S)\\
          &=I(Y;U), \label{eq:reduction}
\end{align}
where \eqref{eq:reduction} follows from the fact that $Y$ is a deterministic
function of  $S$ ($I(Y;U|S)=0$) and  $S\rightarrow Y \rightarrow U$
($I(S;U|Y)=0$). For the additive noise case, the result follows by observing that $H(Y|U)=H(Z)$.
\end{proof}

%%%%%%%%%%%%%%
\subsection{Maximum information leakage} 

The minimum over all possible  maximum cost gains of an adversary that uses a log-loss function in \eqref{eq:maxcostgen} is given by
\begin{equation*} 
C^* =  \max_{u\in \calU}H(S)-H(S|U=u).
\end{equation*}
The previous expression motivates the definition of \textit{maximum information leakage}, presented below.
\begin{defn}
The \textit{maximum information leakage} of a set of features $S$ is defined as the maximum cost gain, given in terms of the log-loss function, that an adversary obtains by observing a single output, and is given by $\max_{u\in \calU}H(S)-H(S|U=u)$. A privacy-preserving mapping $p_{U|Y}(\cdot)$ is said to achieve the \textit{minmax  information leakage} for a distortion constraint $\Delta$ if it is a solution of the minimization
\begin{align} 
  \min_{p_{U|Y}}\max_{u\in\mathcal{U}}~&  H(S)-H(S|U=u)
  \label{eq:innerOpt}\\
  \mbox{s. t. }& \mathbb{E}[d(U,Y)]\leq \Delta \label{eq:distConstrinnerOpt}
\end{align}
\end{defn}

The following theorem demonstrates  how the mapping that achieves the minmax information leakage can be determined as the solution of a related convex program that finds the minimum distortion given a constraint on the maximum information leakage.
\begin{thm}
\label{prop:minmaxleak}
Given $p_{S,Y}(\cdot,\cdot)$, a distortion function $d(\cdot,\cdot)$ and a constraint $\epsilon$ on the maximum information leakage, the minimum achievable distortion and the mapping that achieves the minmax  information leakage can be found by solving the following convex optimization (assuming the implicit simplex constraints on the probability distributions):
\begin{align} 
\min_{p_{U|Y},p_{U|S}} ~& \sum_{u\in\calU} \sum_{ s \in \calS} p_{U|Y}(u|y)p_Y(y)d(u,y)\\
\mbox{\normalfont s.t.~}&  \sum_{y \in \calY} p_{Y|S}(y|s)p_{U|Y}(u|y) = p_{U|S}{(u|s)} ~\forall u,s,\\
&  \sum_{s\in \calS} p_{U|S}(u|s)p_{S}(s) = p_{U}{(u)}~\forall u, \label{eq:constr_pu_minmax}
\end{align}
\vspace{-0.2in}
\begin{align} 
 \delta p_U(u)+ \sum_{s\in \calS} p_{U,S}(u,s)\log
 \frac{p_{U,S}(u,s)}{p_U(u)}&\leq 0~ \forall u,\label{eq:constr_convexMinMax}
\end{align}
where $\delta =  H(S)-\epsilon$. Therefore, for a given value of $\Delta$, the optimization problem in \eqref{eq:innerOpt} can  be efficiently solved with arbitrarily large precision by performing a line-search over  $\epsilon\in [0,H(S)]$ and solving the previous convex program at each step of the search.
\end{thm}
\begin{proof}
The convex program in \eqref{eq:innerOpt} can be reformulated to return the minimum distortion for a given constraint $\epsilon$ on the minmax information leakage as
\begin{align} 
\min_{p_{U|Y}}~& \mathbb{E}[d(U,Y)]\\
\mbox{\normalfont s.t.~}&  H(S|U=u)\geq \delta~\label{eq:constr_minmaxPf}.
\end{align}
It is straightforward to verify that constraint \eqref{eq:constr_convexMinMax} can be written as \eqref{eq:constr_minmaxPf}.  Following the same steps as the proof of Theorem \ref{prop:avgInf} and noting that the function  $g_2(x,z,a) = ax\log(ax/z)$ is convex for $a,x\geq 0$, $z>0$, it follows that  \eqref{eq:constr_minmaxPf} and, consequently,  \eqref{eq:constr_convexMinMax}, is  a convex constraint. Finally, since the optimal distortion value in the previous program is a  decreasing function of $\epsilon$, it follows that the solution of \eqref{eq:innerOpt} can be found through a line-search in $\epsilon$.
\end{proof}

\begin{remark}
Analogously to the average information leakage case, the convex program presented in Theorem \eqref{prop:minmaxleak} can be extended  to the setting where the privacy preserving mapping is given by $p_{U|S}(\cdot|\cdot)$ directly. This can be done by substituting \eqref{eq:constr_pu_minmax} by \eqref{eq:new_constr1} and adding the linear constraint \eqref{eq:new_constr2}.
\end{remark}

Even though the convex program presented in Theorem \ref{prop:minmaxleak} holds in general, it does not provide much insight on the structure of the  privacy mapping that minimizes the maximum information leakage for a given distortion constraint. In order to shed light on the nature of the optimal solution, we present the following result for the particular case when $Y$ is a deterministic function of $S$ and $S\rightarrow Y \rightarrow U$.

\begin{cor}
For $Y=f(S)$, where $f:\calS\rightarrow \calY$ is a deterministic function,  $S\rightarrow Y \rightarrow U$ and a fixed prior $p_{Y,S}(\cdot,\cdot)$, the privacy preserving mapping that minimizes the maximum information leakage is given by
\begin{align}
p_{U|Y}^* = \arg  \min_{p_{U|Y}}~& \max_{u\in \calU} D(p_{Y|U}||\zeta)\\
 \mbox{\normalfont s.t.~}& \mathbb{E}[d(U,Y)]\leq \Delta, \nonumber
\end{align}
where
$\zeta(y) = \frac{ 2^{H(S|Y=y)}  }{ \sum_{y'\in \calY} 2^{H(S|Y=y')} } $.

\end{cor}
\begin{proof}
Under the assumptions of the corollary, note that for a given $u\in \calU$ (and assuming that the logarithms are in base 2)
\begin{align} 
 &H(S|U=u) = \nonumber \\
&   -\sum_{s \in\calS} p_{S|U}(s|u)\log
 p_{S|U}(s|u) \nonumber\\
 &= -\sum_{s \in\calS}\left(\sum_{y \in\mathcal{Y}}
 p_{S|Y}(s|y)p_{Y|U}(y|u)\right) \nonumber\\
& \times \left( \log \sum_{y'\in \mathcal{Y}}p_{S|Y}(s|y') p_{Y|U}(y'|u) \right) \nonumber\\
&= -\sum_{s \in\calS}
 p_{S|Y}(s|f(s))p_{Y|U}(f(s)|u) \nonumber \\
 & ~ ~\times \log p_{S|Y}(s|f(s)) p_{Y|U}(f(s)|u) \label{eq:weirdstep1}\\
&= -\sum_{s \in\calS,y\in \calY}
p_{S|Y}(s|y)p_{Y|U}(y|u) \log p_{S|Y}(s|y) p_{Y|U}(y|u) \label{eq:weirdstep2} \\
& = H(Y|U=u)+\sum_{y \in\mathcal{Y}}
 p_{Y|U}(y|u)H(S|Y=y)\\
& = \sum_{y \in\mathcal{Y}} p_{Y|U}(y|u)\log
 \frac{2^{H(S|Y=y)}}{p_{Y|U}(y|u)} \\
 & = - D(p_{Y|U}||\zeta) + \log\left({\sum_{y\in \calY} 2^{H(S|Y=y)} }\right)
 \label{eq:zeta_rock},
 %&=\sum_{y \in\mathcal{Y}} \frac{p_{U|Y}(u|y)p_Y(y)}{p_U(u)}\log
% \frac{2^{H(X|Y=y)}p_U(u)}{p_{U|Y}(u|y)p_Y(y)}
 %&=\log \norm{\mathbf{p}_{uY}}_1+\frac{1}{\norm{\mathbf{p}_{uY}}_1}\left(\mathbf{h}^T
 %\mathbf{p}_{uY} + H(\mathbf{p}_{uY})\right),
 \end{align}
 where \eqref{eq:weirdstep1} and \eqref{eq:weirdstep2} follows by noting that
$p_{S|Y}(s|y)=0$ if $y\neq f(s)$. The
result follows directly by substituting \eqref{eq:zeta_rock} in
\eqref{eq:innerOpt}.
\end{proof}

For $Y$ a deterministic function of $S$, the optimal privacy preserving mechanism is the one that approximates  (in terms of KL-divergence)  the posterior distribution of $Y$ given $U$ to $\zeta(\cdot)$. Note that the distribution $\zeta(\cdot)$ captures the inherent uncertainty that exists in the function $f$ for different outputs $y\in \calY$. The purpose of the privacy preserving mapping is then to augment this uncertainty, while still satisfying the distortion constraint.  In particular, the larger the uncertainty $H(S|Y=y)$, the larger the probability of $p_{Y|U}(y|u)$ for all $u$. Consequently, the optimal privacy mapping (exponentially) reinforces the posterior probability of the values of $y$ for which there is a large uncertainty regarding the features $S$. This fact is illustrated in the next example, where we revisit  the counting query  presented in Example \ref{examp:query1}.

\begin{example}[Counting query continued]
\label{examp:query2}
Assume that each database input $S_i$, $1\leq i \leq n$ satisfies
$\Pr(\mathbbm{1}_A(S_i)=1)=p$ and are independent and identically distributed.
Then $Y$ is a binomial random variable with parameter $(n,p)$. It follows that
$H(\bS|Y=y) = \log \binom{n}{y}$. Consequently, the optimal privacy preserving
mapping will be the one that results in a posterior probability $p_{Y|U}(y|u)$
that is proportional to the size of the pre-image of $y$, i.e. $p_{Y|U}(y|u)
\propto |f^{-1}(y)|=\binom{n}{y}$.
\end{example}

\section{Comparison of privacy metrics}
\label{sec:compare}

We now compare average information leakage and maximum information leakage with differential privacy and \textit{information privacy}, the latter being a new metric introduced in this section. We first recall the definition of differential privacy, presenting it  in terms of the model discussed in Section \ref{sec:genSet} and assuming that the set of features $\mathbf{S}$ is a vector given by $\mathbf{S}=(S_1,\dots,S_n)$, where $S_i \in \calS$.

\begin{defn}[\cite{dwork_differential_2006}]
  A privacy preserving mapping $p_{U|\bS}(\cdot|\cdot)$ provides $\epsilon$-differential privacy if
  for all inputs $\bs_1$ and $\bs_2$ differing in at most one entry and all
   $B\subseteq \calU$,
   \begin{equation}
     \label{eq:diffPrivacy}
     \PR (U\in B|\bS=\bs_1)\leq \exp(\epsilon)\times
     \PR(U\in B|\bS=\bs_2)~.
   \end{equation}
\end{defn}

An alternative (and much stronger) definition of privacy, related to the one
presented in \cite{evfimievski_limiting_2003} is given below. We note that this definition is unwieldy, but explicitly captures the ultimate goal in privacy: the posterior and prior probabilities of the features $S$ do not change significantly given the output.

\begin{defn} A privacy preserving mapping $p_{U|\bS}(\cdot|\cdot)$  provides
  $\epsilon$-\textit{information privacy} if for all $\bs\subseteq \calS^n$:
  \begin{equation}
  \label{eq:infPrivacy}
    \exp(-\epsilon)\leq \frac{p_{\bS|U}(\bs|u)}{p_{\bS}(\bs)} \leq \exp(\epsilon)~\forall u\in \mathcal{U} : p_U(u)> 0.
   \end{equation}
 \end{defn}
Note that $\epsilon$-information privacy implies directly $2\epsilon$-differential
privacy and maximum information leakage of at most $\epsilon/\ln 2$ bits, as shown below.

\begin{thm}
  If a privacy preserving mapping  $p_{U|\bS}(\cdot|\cdot)$ is $\epsilon$-information private for
  some input distribution such that $\mbox{supp}(p_U)=\mathcal{\calU}$ , then
  it is at least $2\epsilon$-differentially private and leaks at most $\epsilon/\ln 2$ bits on average.
\end{thm}
\begin{proof}
Note that for a given $B\subseteq \calU$
\begin{align*}
  \frac{\PR(U \in B|\bS=\bs_1)}{\PR(U \in B|\bS=\bs_2)}&=
  \frac{\PR(\bS=\bs_1|U \in B)\PR(\bS=\bs_2)}{\PR(\bS=\bs_2|U \in B)\PR(\bS=\bs_1)}\\
  &\leq \exp(2\epsilon),
\end{align*}
where the last step follows from \eqref{eq:diffPrivacy}. Clearly if $\bs_1$ and
$\bs_2$ are neighboring vectors (i.e. differ by only one entry), then $2\epsilon$-differential privacy is
satisfied. Furthermore
\begin{align*}
%I(\bS;U) &= \sum_{\bs\in \calS^n, u\in \calU} p_{\bS|U}(\bs|u)p_U(u)\log \frac{p_{\bS|U}(\bs|u)}{p_\bS(\bs)}\\
H(\bS)-H(\bS|U=u)& = \sum_{\bs\in \calS^n} p_{\bS|U}(\bs|u)p_U(u)\log \frac{p_{\bS|U}(\bs|u)}{p_\bS(\bs)}\\
&\leq \sum_{\bs\in \calS^n, u\in \calU}
p_{\bS|U}(\bs|u)p_U(u)\frac{\epsilon}{\ln 2}\\
&= \frac{\epsilon}{\ln 2}
\end{align*}\end{proof}
\vspace{-0.05in}

%
%However, $\epsilon$-differential privacy does not provide \textit{in general}
%information privacy, as shown in the following proposition.
%
%\begin{prop}
%  There exists probabilistic function $f:\calS^n\rightarrow \calY$ and mapping $p_{U|Y}$ that is $\epsilon$-differentially
%  private, but for every $\delta>0$ there exists a prior distribution $p_S$ over
%  the input data sets such that $p_{U|Y}$ violates $\delta$-information
%  privacy.
%\end{prop}
%\begin{proof}
%(COMPLETE)
%\end{proof}

We show in the next theorem that differential privacy \textit{does not
guarantee} privacy in terms of average information leakage \textit{in general}
and, consequently in terms of maximum information leakage and information
privacy. More specifically, guaranteeing that a mechanism is
$\epsilon$-differentially private \textit{does not} provide \textit{any}
guarantee on the information leakage.
\begin{thm}
\label{prop:diffPrivSucks}
  For every $\epsilon >0$ and $\delta\geq 0$, there exists an $n\in \mathbb{Z}_{+}$, sets $\calS^n$ and $\calU$, a prior $p_\bS(\cdot)$ over $\calS^n$ and a privacy mapping $p_{U|S}(\cdot|\cdot)$ that is $\epsilon$-differentially
  private but leaks at least $\delta$ bits on average.
\end{thm}
\begin{proof}
We prove the statement by explicitly constructing an example that is
$\epsilon$-differentially private, but an arbitrarily large amount of
information can leak on average from the system. For this, we return to the
counting query  discussed in examples \ref{examp:query1} and  \ref{examp:query2}
with, the sets $\mathcal{S}$ and $\calY$  being defined accordingly, and letting
$\calU=\calY$. We do not
assume independence of the inputs.

For the counting query and for any given prior, adding Laplacian noise to the output provides $\epsilon$-differential privacy \cite{dwork_differential_2006}. More precisely, for the output of the query given in \eqref{eq:count_query}, denoted  as $Y\sim p_Y(y),0\leq y \leq n$, the mapping
\begin{equation}
  U=Y+N,~~N\sim \mbox{Lap}(1/\epsilon), \label{eq:lap_priv}
\end{equation}
where the pdf of the additive noise $N$ given by
\begin{equation}
  p_N(r;\epsilon)=\frac{\epsilon}{2}\exp(-|r|\epsilon),
\end{equation}
is $\epsilon$-differentially private.
Now assume that $\epsilon$ is given, and denote $\bS = (X_1,\dots,X_n)$. Set $k$
and $n$ such that $n~\mod k = 0$, and let $p_{\bS}(\cdot)$ be such that
\begin{equation}
p_Y(y) = \left\{
     \begin{array}{ll}
       \frac{1}{1+n/k} & \mbox{if~} y\mod k = 0,\\
       0 & \mbox{otherwise.}
     \end{array} \right.
\end{equation}
 With the goal of lower-bounding the information leakage, assume that Bob, after
 observing $U$, maps it to the nearest value of $y$ such that $p_Y(y)>0$, i.e.
 does a maximum a posteriori estimation of $Y$.  The probability that Bob makes
 a correct estimation (and neglecting edge effects), denoted by
 $\alpha_{k,n}(\epsilon)$, is given by:
\begin{equation}
 \alpha_{k,n}(\epsilon) = \int_{\frac{-k}{2}}^{\frac{k}{2}} \frac{\epsilon}{2}\exp(-|x|\epsilon) dx=
  1-\exp\left(-\frac{k\epsilon}{2} \right).
\end{equation}
Let $E$ be a binary random variable that indicates the event that Bobs makes a wrong estimation of $Y$ given $U$. Then
\vspace{-0.05in}
\begin{align*} 
I(Y;U)&\geq I(E,Y;U)-1\\
&\geq I(Y;U|E)-1\\
%I(Y;U)&\geq I(Y;U|E)-1\\
&\geq \PR\{E=0\}I(Y;U|E=0)-1\\
&= \left(1-e^{-\frac{k\epsilon}{2} }\right)\log\left(1+\frac{n}{k}\right)-1,
\end{align*}
which can be made arbitrarily larger than $\delta$ by appropriately choosing the
values of $n$ and $k$. Since $Y$ is a deterministic function of $\bS$, $I(Y;U) =
I(\bS;U)$, as shown in the proof of Corollary \ref{avg:determ}, and the result
follows.  
\end{proof}

The counterexample used in the proof of the previous theorem can
  be extended to allow the adversary to recover \textit{exactly} the inputs
  generated the ouput $U$. This can be done by assuming that the inputs are ordered and
  correlated in such a way that $Y=y$ if and only if $S_1=1,\dots,S_y =1$. 
  In this case, for $n$ and $k$ sufficiently large,
  the adversary can exploit the input correlation to correctly learn the values of
  $S_1,\dots,S_n$ with arbitrarily high probability.

Differential privacy does not necessarily guarantee low leakage of information
-- in fact, an arbitrarily large amount of information can be leaking from a
differentially private system, as shown in Theorem \ref{prop:diffPrivSucks}.
This is a serious issue when using solely the differential privacy definition as
a privacy metric. In addition, it follows as a simple extension of \cite[Prop.
4.3] {McGregor-ECCC2011} that $I(S;U)\leq O(\epsilon n)$, corroborating that
differential privacy does not bound above the average information leakage when
$n$ is sufficiently large.

%Furthermore, average information leakage and  maximum information leakage do not necessarily imply differential privacy. However, this \textit{is not} due to the fact that differential privacy is a more stringent privacy metric in certain settings.

Nevertheless, differential privacy does have an operational advantage since it
does not require any prior information. However, by neglecting the prior and
requiring differential privacy, the resulting  mapping might not be \textit{de
facto} private, being  suboptimal under the information leakage measure. We note
that the presented formulations can be made prior independent 
maximizing the minimum information leakage over a set of possible priors. This
problem is closely related to universal coding \cite{cover_elements_2006}.
%\begin{remark}
% Let $p^*_{U|Y}(u|y;p_{S,Y})$ and $\tilde{p}^*_{U|Y}(u|y;p_{S,Y})$  represent the optimal privacy mapping that minimizes the average and the maximum information leakage, respectively, for a given prior  $p_{S,Y}(\cdot,\cdot)$ and certain distortion constraints. Furthermore, let $\calP_{S,Y}$ represent a set of feasible priors for $S$ and $Y$. Then a prior-independent privacy preserving mapping $q^*_{U|Y}(u|y;p_{S,Y})$ and $\tilde{q}^*_{U|Y}(u|y;p_{S,Y})$ that minimizes the worst-case average information  leakage and maximum information leakage over and  $\calP_{S,Y}$, respectively, can be found by solving
%\begin{equation}
%\min_{p_{S,Y} \in\calP_{S,Y} } q^*_{U|Y}(u|y;p_{S,Y}) \mbox{ or } \tilde{q}^*_{U|Y}(u|y;p_{S,Y})~.
%\ end{equation}
%\end{remark}
%In particular, for the worst-case average information leakage, this problem is
%closely related to that of universal source coding
%\cite{cover_elements_2006}.
\vspace{-0.04in}

\section{Conclusions}\label{sec:Conclusion}
In this paper we presented a general statistical inference framework to capture the privacy threat incurred by a user that releases data to a passive but curious adversary given utility constraints. We demonstrated how under certain assumptions this framework naturally leads to an information-theoretic approach to privacy. The design problem of finding privacy-preserving mappings for minimizing the information leakage from a user's data with utility constraints was formulated as a convex program. This approach can lead to practical and deployable privacy-preserving mechanisms. Finally, we compared our approach with differential privacy, and showed that the differential privacy requirement does not necessarily constrain the information leakage from a data set.

\vspace{-0.05in}

\bibliographystyle{IEEEtran}
\bibliography{IEEEabrv,references}

\end{document}